\documentclass[a4paper]{article}

\usepackage[latin1]{inputenc} %
\usepackage[T1]{fontenc} %
\usepackage{RR}
\usepackage{hyperref}
\usepackage{latexsym}
\usepackage{amssymb,amsmath}
\usepackage{amsthm}
\usepackage{graphicx}
\usepackage{color}

\newcommand{\R}{\mathbb{R}}

\newtheorem{theorem}{Theorem}
\newtheorem{lemma}[theorem]{Lemma}
\newtheorem{cor}[theorem]{Corollary}

\newcommand{\red}[1]{\textcolor{red}{#1}}

\RRdate{December 2009}

\RRauthor{
Otfried Cheong\thanks{Theory of Computation Lab, Dept. of Computer Science, KAIST, Korea. \texttt{otfried@kaist.edu}} \and
Xavier Goaoc\thanks{Loria, INRIA Nancy Grand Est, France. \texttt{goaoc@loria.fr}} \and
Cyril Nicaud\thanks{LIGM Universit\'e Paris Est, France.  \texttt{nicaud@univ-mlv.fr}}
}
\authorhead{Cheong, Goaoc \& Nicaud}

\RRtitle{Set Systems and Families of Permutations with Small Traces}
\RRetitle{Set Systems and Families of Permutations with Small Traces}
\titlehead{Set Systems and Families of Permutations with Small Traces}

\RRnote{Otfried Cheong was supported by the Korea Science and Engineering Foundation Grant R01-2008-000-11607-0 funded by the Korea government. The collaboration between Otfried Cheong and Xavier Goaoc was supported by the INRIA \emph{\'Equipe associ\'ee} KI.}

\RRabstract{We study the maximum size of a set system on $n$ elements
 whose trace on any $b$ elements has size at most $k$. We show that if
 for some $b \ge i \ge 0$ the shatter function $f_R$ of a set system
 $([n],R)$ satisfies $f_R(b) < 2^i(b-i+1)$ then $|R| = O(n^i)$; this
 generalizes Sauer's Lemma on the size of set systems with bounded
 VC-dimension. We use this bound to delineate the main growth rates
 for the same problem on families of permutations, where the trace
 corresponds to the inclusion for permutations. This is related to a
 question of Raz on families of permutations with bounded VC-dimension
 that generalizes the Stanley-Wilf conjecture on permutations with
 excluded patterns.}
\RRkeyword{Set systems, VC dimension, Sauer Lemma, Permutation pattern}

\RRresume{Nous \'etudions la taille maximale d'un hypergraphe \`a $n$
 sommets dont la trace sur toute sous-famille de $b$ sommets est de
 taille au plus $k$. Nous montrons que pour tous entiers $b \ge i \ge
 0$, si la fonction de pulv\'erisation $f_R$ d'un hypergraphe
 $([n],R)$ satisfait $f_R(b) < 2^i(b-i+1)$ alors $|R| = O(n^i)$; cela
 g\'en\'eralise le Lemme de Sauer sur la taille des hypergraphes de
 dimension de Vapnik-Chervonenkis born\'ee. Nous utilisons ensuite
 cette borne pour s\'eparer les principaux r\'egimes de croissance
 pour une question analogue sur les familles de permutations, o\`u
 l'op\'eration de trace correspnd \`a l'inclusion de motifs. Cela est
 reli\'e \`a une question de Raz sur les familles de permutations \`a
 dimension de Vapnik-Chervonenkis born\'ee qui g\'en\'eralise la
 conjecture de Stanley-Wilf sur les permutations \`a motifs exclus.}
\RRmotcle{Hypergraphes, Dimension de Vapnik-Chervonenkis, Lemme de Sauer, Motifs de permutation}

\RRprojet{V\'egas}
\RRdomaine{2}
\RRtheme{Algorithmique, calcul certifié et cryptographie}

\URLorraine
\RCNancy

\begin{document}
\RRNo{7154}

\makeRR



\section{Introduction}

 In this paper, we study two problems of the following flavor: how
 large can a family of combinatorial objects defined on
 $[n]=\{1,\ldots, n\}$ be if its number of distinct ``projections'' on
 any small subset is bounded? We consider set systems, where the
 ``projection'' is the standard notion of trace, and families of
 permutations, where the ``projection'' corresponds to the notion of
 inclusion used in the study of permutations with excluded patterns.

\paragraph{\bf Set systems.} A \emph{set system}, also called a
 \emph{range space} or a \emph{hypergraph}, is a pair $(G,R)$ where
 $G$ is a set, the \emph{ground set}, and $R$ is a set of subsets of
 $G$, the \emph{ranges}. Since we will only consider finite set
 systems, our ground set will always be $[n]$. Given $X \subset [n]$,
 the \emph{trace} of $R$ on $X$, denoted $R_{|X}$, is the set $\{ A
 \cap X \mid A \in R\}$. Given an integer $b$, let $\binom{R}{b}$ denote
 the set of $b$-tuples of $R$, and define:
\[ f_R(b) = \max_{X \in \binom{[n]}{b}} |R_{|X}|.\]
 The function $f_R$ is called the \emph{shatter function} of
 $([n],R)$, and counts the size of the largest trace on a subset of
 $[n]$ of size $b$. The first problem we consider is the following:

\medskip

\begin{quote}
\noindent
\textbf{Question 1.} Given $b$ and $k$, how large can a set system $([n],R)$ be if $f_R(b)\le k$?
\end{quote}

\medskip

\noindent
 For $k=2^b-1$, the answer is given by Sauer's Lemma~\cite{Sauer71} (also
 proven independently by Perles and Shelah~\cite{Shelah72} and Vapnik
 and Chervonenkis~\cite{VC71}), which states that:
\begin{equation}\label{eq:sauer} |R| \le \sum_{i=0}^{b-1} \binom{n}{i} = O(n^{b-1}).\end{equation}
 The largest $b$ such that $f_R(b)=2^b$ is known as the
 \emph{VC-dimension} of $([n],R)$. The theory of set systems of
 bounded VC-dimension, and in particular Sauer's Lemma, has many
 applications, in particular in geometry and approximation algorithms;
 classical examples include the epsilon-net Theorem~\cite{Discrepancybook} or
 improved approximation algorithms for geometric set cover~\cite{BronnimannG-95}.

\medskip

 For the case of graphs, that is, set systems where all ranges have
 size $2$, Question~1 is a classical problem known as a
 \emph{Dirac-type problem}: what is the maximum number $Ex(n,m,k)$ of
 edges in a graph on $n$ vertices whose induced subgraph on any $m$
 vertices has at most $k$ edges? These problems were extensively
 studied in extremal graph theory since the 1960's, and we refer to
 the survey of Griggs et al.~\cite{gst-ext98} for an overview. In the
 case of general set systems, the only results we are aware of are due
 to Frankl~\cite{Frankl83} and Bollob\'as and
 Radcliffe~\cite{DefectSauer}. Specifically, Frankl proved that
\[ f_R(3) \le 6 \Rightarrow |R| \le t_2(n) +n +1
 \quad \hbox{and} \quad f_R(4) \le 10 \Rightarrow |R| \le t_3(n) +n+1,
\]
 where $t_i(n)$ denotes the number of edges of the Tur\'an graph
 $T_i(n)$. Bollob\'as and Radcliffe showed that:
\[ f_R(4) \le 11 \Rightarrow |R| \le \binom{n}2+n+1 \quad \hbox{except for $n=6$.}\]
 There has also been interest in the case where $b=\alpha n$ and
 $b=n-\Theta(1)$; we refer to the article of Bollob\'as and
 Radcliffe~\cite{DefectSauer} for an overview of these results.

\paragraph{\bf Permutations.} The notion of VC-dimension was extended to
 sets of permutations by Raz~\cite{Raz00vc-dimensionof} as
 follows. Let $\sigma$ be a permutation on $[n]$ and $X$ some subset
 of $[n]$. The \emph{restriction} of $\sigma$ to $X$ is the
 permutation $\sigma_{|X}$ of $X$ such that for any $u,v \in X$,
 $\sigma_{|X}^{-1}(u)<\sigma_{|X}^{-1}(v)$ whenever $\sigma^{-1}(u) <
 \sigma^{-1}(v)$; if we consider a permutation as an ordering,
 $\sigma_{|X}$ is simply the order induced on $X$ by $\sigma$. This
 allows to define the \emph{shatter function} of a set $F$ of permutations
 similarly:
\[ \phi_F(m) = \max_{X \in \binom{[n]}{m}} |F_{|X}|.\]
 The VC-dimension of $F$ is then the largest $m$ such that $\phi_F(m)
 = m!$, and the analogue of Question~1 arises naturally for sets of
 permutations:

\medskip

\begin{quote}
\noindent
\textbf{Question 2.} Given $m$ and $k$, how large can a set $F$ of
 permutations on $[n]$ be if $\phi_F(m)\le k$?
\end{quote}

\medskip

\noindent
 Raz~\cite{Raz00vc-dimensionof} showed that any family of permutations
 on $[n]$ such that $\phi_F(3)<6$ has size at most exponential in $n$,
 and asked whether the same holds whenever $k<m!$.

\medskip

 This problem is related to classical questions on families of
 permutations with \emph{excluded pattern}.  A permutation $\sigma$ on
 $[n]$ \emph{contains} a permutation $\tau$ on $[m]$ if there exists
 $a_1 < a_2 < \ldots < a_m$ in $[n]$ such that $\sigma^{-1}(a_i) <
 \sigma^{-1}(a_j)$ whenever $\tau^{-1}(i)<\tau^{-1}(j)$. If no
 permutation in a family $F$ contains $\tau$ then $F$ \emph{avoids}
 $\tau$ and $\tau$ is an \emph{excluded pattern} for $F$. The study of
 families of permutations with excluded patterns goes back to a work
 of Knuth~\cite{AOCP1}, motivated by sorting permutations using
 queues, and received considerable attention over the last decades. In
 particular, Stanley and Wilf asked whether for any fixed permutation
 $\tau$ the number of permutations on $[n]$ that avoid $\tau$ is at
 most exponential in $n$, a question answered in the positive by
 Marcus and Tardos~\cite{MarTar:SW(2004)}. If a family of permutations
 has VC-dimension at most $m-1$ then for any $m$-tuple $X \subset [n]$
 there is a permutation $\sigma(X)$ on $[m]$ which is forbidden for
 restrictions \emph{to $X$}. In that sense, Raz's question generalizes
 that of Stanley and Wilf.

\paragraph{\bf Our results.} In this paper, we generalize Sauer's Lemma,
 and show that for any range space $([n],R)$, if $f_R(b)<2^i(b-i+1)$
 for some $b > i \ge 0$ then $|R| = O(n^i)$
 (Theorem~\ref{thm:ranges-upper}). We then prove that the condition
 $f_R(b) = k$ is in fact \emph{equivalent} to a Dirac-type problem on
 graphs for $k \le 8+3\lfloor \frac{b-3}2 \rfloor + s(b)$, where
 $s(b)=1$ when $b$ is even and $0$ otherwise
 (Lemma~\ref{lem:red-to-dirac}). It follows that some conditions
 $f_R(b)=k$ lead to growth rates with fractional exponents
 (Corollary~\ref{cor:48}), a behavior not captured by
 Theorem~\ref{thm:ranges-upper}. Finally, we give a reduction of the
 permutation problem to the set system problem
 (Lemma~\ref{lem:correspondance}) from which we deduce the main
 transitions between the constant, polynomial and at least exponential
 behaviors for Question~2.

\section{Set systems}

 In this section we give bounds on the size of a set $R$ of ranges on
 $[n]$ with a given $f_R(b)$. Recall that a set system $([n],R)$ is
 \emph{ideal}, also called \emph{monotone decreasing}\footnote{An
 ideal set system is also an \emph{abstract simplicial complex} to
 which the empty set was added.}, if for any $B \subset A \in R$ we
 have $B \in R$. The next lemma was proven, independently, by
 Alon~\cite{Alon83} and Frankl~\cite{Frankl83}.

\begin{lemma}\label{lem:normalization}
 For any set system $([n],R)$ there exists an ideal set system
 $([n],\tilde{R})$ such that $|R| = |\tilde{R}|$ and for any integer
 $b$ we have $f_{\tilde{R}}(b) \le f_R(b)$.
\end{lemma}

\noindent
 This can be shown by defining, for any $x\in[n]$, the operator (also called a \emph{push-down} or a \emph{compression})
\[ \tilde{T}_x(R) = \{ A\setminus\{x\} \mid A \in R\} \cup \{ A \mid A \in R \hbox{ such that } x \in A \hbox{ and } A\setminus\{x\} \in R\},\]
 that removes $x$ from any range in $R$ where that does not decrease
 the total number of sets. Then,
\[ \tilde{R} = \tilde{T}_1\left(\tilde{T}_2\left(\ldots (\tilde{T}_n\left(R\right)\right) \ldots \right)\]
 is one such ideal set. We refer to Bollob\'as~\cite[Chapter
 17]{bollobas86} and the survey of F\"uredi and
 Pach~\cite{fp-survey-91} for more details. An immediate consequence
 of Lemma~\ref{lem:normalization} is that we can work with ideal set
 systems when studying our first question.

\subsection{Sauer's Lemma for small traces}

 Define $\binom{n}{-1} = 0$ and consider the sequence $\upsilon_i(b) =
 2^i (b-i+1)$ that interpolates between $b+1=\upsilon_0(b)$ and $2^b =
 \upsilon_{b-1}(b)$. Our first result is the following generalization of
 Sauer's Lemma.

\begin{theorem}\label{thm:ranges-upper}
 Let $b > i \ge 0$ be two integers. Any range space $([n],R)$ with
 $f_R(b) < \upsilon_i(b)$ has size $|R| = f_R(n) < \sum_{j=0}^{i} (b-j+1)\binom{n}{j}$
\end{theorem}
\begin{proof}
 By Bondy's Theorem~\cite{bollobas86}, for any $b+1$ distinct ranges
 there exist $b$ elements on which they have distinct trace. It
 follows that if $f_R(b)<b+1$ we also have $f_R(n)<b+1$ for any $n$,
 and the statement holds for $i=0$. Also, from
\[ \sum_{j=0}^{i} (b-j+1)\binom{b}{j} \ge (b-i+1)\sum_{j=0}^{i} \binom{b}{j} \ge (b-i+1)2^i = \upsilon_i(b),\]
 we have that the statement holds for $n=b$ and any $i$.

\medskip

 Now, we fix $b$ and assume that we have
\[ f_R(b) < \upsilon_k(b) \quad \Rightarrow \quad f_R(t) < \sum_{j=0}^{k} (b-j+1)\binom{t}{j}\]
 whenever $k<i$ or $k=i$ and $t<n$. Let $R' = R_{|[n-1]}$
 denote the trace of $R$ on $[n-1]$ and let $D$ denote the ranges in
 $R'$ that are the trace of two distinct ranges from $R$. Notice that:
\begin{equation}\label{eq:Rp}
 |R| = |R'| + |D| \quad \hbox{and} \quad f_{R'}(b)<\upsilon_i(b).
\end{equation}
 Since $D \subset R'$, we have that $|D_{|X}| \le |R'_{|X}|$ and thus
 $|D_{|X}| \le \frac12 |R_{|(X \cup \{n\})}|$. It follows that
 $f_D(b-1) \le \left\lfloor\frac{f_R(b)}2\right\rfloor$. Now, from
 $\upsilon_i(b) = 2\upsilon_{i-1}(b-1)$ we get that:
\begin{equation}\label{eq:D}
 f_{D}(b-1)<\upsilon_{i-1}(b-1).
\end{equation}
 From Equations~(\ref{eq:Rp}) and~(\ref{eq:D}) and the induction hypothesis we obtain:
\[ |R| < \sum_{j=0}^{i} (b-j+1)\binom{n-1}{j} + \sum_{j=0}^{i-1} (b-1-j+1)\binom{n-1}{j}.\]
 This rewrites as
\[ |R| < b+1+\sum_{j=1}^{i} (b-j+1)\binom{n-1}{j} + \sum_{j=1}^{i} (b-j+1)\binom{n-1}{j-1}\]
 and with $\binom{n-1}{j} + \binom{n-1}{j-1} =  \binom{n}{j}$ we get
\[ |R| < b+1+\sum_{j=1}^{i} (b-j+1)\binom{n}{j} = \sum_{j=0}^{i} (b-j+1)\binom{n}{j},\]
 and thus:
\[ f_R(b) < \upsilon_i(b) \quad \Rightarrow \quad f_R(n) < \sum_{j=0}^{i} (b-j+1)\binom{n}{j}.\]
 The statement follows by induction.
\end{proof}

\bigskip
 Now, consider the following family of lower bounds. For $i=1,
 \ldots, b$ let
\[ \lambda_i(b) = \max_{b=b_1+ \ldots b_i}\prod_{j=1}^i (b_j+1)\]
 and consider the system $([n],R)$ where $R$ is obtained by splitting
 $[n]$ into $i$ roughly equal subsets and picking all $i$-tuples
 containing one element from each subset. Notice that $|R| =
 \Omega(n^i)$ and that $f_R(b) \le \lambda_i(b)$. The same holds for
 $i=0$ with $\lambda_0(b)=1$.  Thus, for any $k$ such that
 $\lambda_i(b) \le k < \upsilon_i(b)$, the maximum size of a set
 system $([n],R)$ with $f_R(b)=k$ is $\Theta(n^i)$.

\begin{table}[!htb]
{\small
\hspace{-1.7cm}
\begin{tabular}[7]{|c||c|c||c|c||c|c||c|c||c|c|}
\hline
$b$ & $\upsilon_0(b)$-1 & $\lambda_1(b)$ & $\upsilon_1(b)$-1 & $\lambda_2(b)$ & $\upsilon_2(b)-1$ & $\lambda_3(b)$ & $\upsilon_3(b)-1$ & $\lambda_4(b)$ & $\upsilon_4(b)-1$ & $\lambda_5(b)$ \\
$|R|$ & $O(1)$ & $\Omega(n)$ & $O(n)$ & $\Omega(n^2)$ & $O(n^2)$ & $\Omega(n^3)$ & $O(n^3)$ & $\Omega(n^4)$ & $O(n^4)$ & $\Omega(n^5)$\\
\hline
\hline
$2$ & $2$ & $3$ & &&&&&&& \\
\hline
$3$ & $3$ & $4$ & $5$ & $6$ & &&&&&  \\
\hline
$4$ & $4$ & $5$ & \red{$7$} & \red{$9$} & $11$ & $12$ & &&& \\
\hline
$5$ & $5$ & $6$ & \red{$9$} & \red{$12$} & \red{$15$} & \red{$18$} & $23$ & $24$ &&\\
\hline
$6$ & $6$ & $7$ & \red{$11$} & \red{$16$} & \red{$19$} & \red{$27$} & \red{$31$} & \red{$36$} & $47$ & $48$ \\
\hline
\end{tabular}
\caption{The values $v_i(b)$ and $\lambda_i(b)$ for small $b$. Gaps appear in red.}
}
\end{table}

 In particular, the order of magnitude given by
 Theorem~\ref{thm:ranges-upper} is tight for all $b\le 4$, with the
 exception of set systems with $f_R(4)=8$.

\paragraph{\bf Remark.} Observe that the condition that $f_R(b) <
 \upsilon_i(b)$ does not imply that $R$ has VC-dimension at most
 $i$. A simple example is given by
\[ R = \{A \mid A \subset [i]\} \cup \{\{x\} \mid x \in [n]\},\]
 which has VC-dimension $i$ and for which $f_R(b) = 2^i+b-i-1$ is
 smaller than $\upsilon_{i-1}(b)=2^{i-1}(b-i)$ for $b$ large enough.

\subsection{Equivalence with Dirac-type problems}

 Recall that $Ex(n,m,k)$ denotes the maximum number of edges in a
 graph on $n$ vertices whose induced subgraph on any $m$ vertices has
 at most $k$ edges.  Let $\zeta(b) = 8+3\lfloor \frac{b-3}2 \rfloor
 + s(b)$ where $s(b)=1$ if $b$ is even and $0$ otherwise.

\begin{lemma}\label{lem:red-to-dirac}
 For any $b \ge 3$, the maximal size of a set system $([n],R)$ with
 $f_R(b) = \zeta(b)-1$ is $Ex(n,b,\zeta(b)-b-2)+n+1$.
\end{lemma}
\begin{proof}
 By Lemma~\ref{lem:normalization} it suffices to prove the statement
 for ideal set systems. Let $([n],R)$ be an ideal set system with
 $f_R(k)<\zeta(b)-1$ and maximal size. If $R$ contains some range $A$
 of size $3$, then $|R_{|A}|=8$. Now, write $b=3+2j+s$ with $s\in
 \{0,1\}$. Let $B$ denote the set $A$ augmented by $j$ pairs of
 elements that belong to $R$, and one single element of $R$ if $s =
 1$. The set $B$ has size $b$ and the trace of $R$ on $B$ has size at
 least $8+3j+s = \zeta(b)$. Thus, if $f_R(b) < \zeta(b)$ we get that
 $R$ contains no triple, and can thus be decomposed into
\[ R = \{\emptyset\} \cup V \cup E,\]
 where $V$ are the singletons and $E$ the pairs in $R$; call the
 former the \emph{vertices} of $R$ and the latter its \emph{edges}. If
 some element $x \in [n]$ is not a singleton of $V$ then it is
 contained in no range of $R$, and we can delete it without changing
 the size of $R$; this contradicts the maximality of $R$. Now, notice
 that the trace of $R$ on any $b$ elements contains at most
 $f_R(b)-b-1 = \zeta(b)-b-2$ edges, since it contains the empty set
 and each of the $b$ vertices. Conversely, let $G=([n],E)$ be a graph
 whose induced graph on any $b$ vertices has at most $\zeta(b)-b-2$
 edges. If $R=\{\emptyset\} \cup [n] \cup E$ then the set system
 $([n],R)$ satisfies $f_R(b) <\zeta(b)$ and the statement follows.
\end{proof}
 
 A graph whose induced subgraphs on any $m$ vertices have at most $k <
 \lfloor \frac{m^2}4 \rfloor$ edges cannot contain a $K_{\lfloor
 \frac{k}2 \rfloor,\lfloor \frac{k}2 \rfloor}$, and thus, by the
 K\H{o}v\'ari-S\'os-Tur\'an Theorem, has at most $Ex(n,m,k) =
 O\left(n^{2-\frac1{\lfloor \frac{k}2 \rfloor}}\right)$ edges. It
 follows that:
\[ Ex(n,4,3) = Ex(n,5,5) = O(n\sqrt{n}).\]
 The classical constructions yielding bipartite graphs on $n$ vertices
 with $\Theta(n\sqrt{n})$ edges and no $K_{2,2}$ show that this bound
 is best possible. Since $\zeta(4)=9$, we get that the family of
 growth rates obtained by the conditions $f_R(b)=k$ does not only
 contain polynomial growth with integer exponents:

\begin{cor}\label{cor:48}
 The largest set system $([n],R)$ with $f_R(4)=8$ or $f_R(5)=10$ has
 size $|R| = \Theta(n\sqrt{n})$.
\end{cor}

\noindent
 Note that Lemma~\ref{lem:red-to-dirac} can be extended into an
 equivalence of Question~1 and Dirac's problem on $r$-regular
 hypergraphs for arbitrary large $r$.

\section{Families of permutations}

 In this section we give bounds on the size of a family $F$ of
 permutations on $[n]$ with a given $\phi_F(b)$.

\paragraph{\bf Reduction to set systems.} An \emph{inversion} of a
 permutation $\sigma$ on $[n]$ is a pair of elements $i<j$ such that
 $\sigma^{-1}(i)>\sigma^{-1}(j)$. The \emph{distinguishing pair} of
 two permutations $\sigma_1$ and $\sigma_2$ is the lexicographically
 smallest pair $(i,j) \subset [n]$ that appears in different orders in
 $\sigma_1$ and $\sigma_2$, i.e. is an inversion for one but not for the
 other. If $F$ is a family of permutations on $[n]$ we let $I_F$
 denote the set of distinguishing pairs of pairs of permutations from
 $F$. Given a permutation $\sigma \in F$, we let $R(\sigma)$ denote
 the set of elements of $I_F$ that are inversions of $\sigma$, and let
 $R(F) = \{R(\sigma)\mid\sigma \in F\}$. Observe that $(I_F,R(F))$ is a
 range space and that $R$ is a one-to-one map between $F$ and
 $R(F)$. In particular $|F| = |R(F)|$.

\begin{lemma}\label{lem:correspondance}
$f_{R(F)}(\lfloor \frac{m}2\rfloor) \le \phi_F(m) $ and $|I_F| \le Ex(n,m,\phi_F(m)-1)$.
\end{lemma}
\begin{proof}
 Consider $b=\lfloor \frac{m}2\rfloor$ elements $(p_1, \ldots, p_b)$
 in $I_F$ and assume there exists $k$ ranges $R(\sigma_1), \ldots,
 R(\sigma_{k})$ with distinct traces on $\{p_1, \ldots, p_b\}$. Then
 the restrictions of $\sigma_1, \ldots, \sigma_{k}$ on $X=\cup_{1 \le
 i \le b}p_i$ must also be pairwise distinct. Thus, $\phi_F(m) \ge k$
 whenever $f_{R(F)}(\lfloor \frac{m}2\rfloor) \ge k$, and the
 statement follows.

 Let $s(t)$ denote the maximum number of distinguishing pairs in a
 family of $t$ permutations (on $[n]$). From
\[ s(2) = 1 \quad \hbox{and} \quad s(t) \le 1+\max_{1 \le i \le t-1} \left\{s(i)+s(t-i)\right\},\]
 we get that $s(t) \le t-1$ by a simple induction. This implies that
 in the graph $G=([n],I_F)$, any $m$ vertices span at most
 $\phi_F(m)-1$ edges, and it follows that
\[|I_F| \le Ex(n,m,\phi_F(m)-1),\]
 which concludes the proof.
\end{proof}

 A subquadratic $I_F$ is not always possible: every pair is a distinguishing pair of the family of all permutations on $[n]$ that restrict to
 the identity on some $(n-1)$-tuple. For that family, $\phi_F(m) =
 (m-1)^2+1$.

\paragraph{\bf Main transitions.} We can now outline the main transitions
 in the growth rate of families of permutations according to the value
 of $\phi_F(m)$. Let $b = \lfloor \frac{m}2 \rfloor$.
\begin{itemize}
\item If $\phi_F(m) \le \lfloor \frac{m}2 \rfloor$ then, by
 Lemma~\ref{lem:correspondance}, $f_{R(F)}(b) \le b$ and
 Theorem~\ref{thm:ranges-upper} with $i=0$ yields that $|F| = |R(F)| =
 O(1)$.

\item Assume that $\lfloor \frac{m}2 \rfloor < \phi_F(m) < 2\lfloor
 \frac{m}2 \rfloor$. Then, by Lemma~\ref{lem:correspondance},
 $f_{R(F)}(b) < 2b$ and Theorem~\ref{thm:ranges-upper} with $i=1$
 yields that $|F| = |R(F)| = O(|I_F|) = O(Ex(n,m,m-2)) =O(n)$. A
 matching lower bound is given by the family
\begin{quote}
 $F_1:$ all permutations on $[n]$ that differ from the identity by
 the transposition of a single pair of the form $(2i,2i+1)$,
\end{quote}
 of size $1+\lfloor \frac{n}2 \rfloor$ and with $\phi_{F_1}(m) = \lfloor \frac{m}2 \rfloor+1$.

\item If $\phi_F(m) < 2^{\lfloor \frac{m}2 \rfloor}$ then, by
 Lemma~\ref{lem:correspondance}, $f_{R(F)}(b) < 2^b$ and $(I_F,R(F))$
 has VC-dimension at most $b-1$. It follows, from Sauer's Lemma, that $|F| = |R(F)| =
 O(|I_F|^{b-1})$, and since $|I_F| = O(n^2)$, we get that $|F|$ is
 $O\left(n^{2\lfloor \frac{m}2 \rfloor-2}\right)$.

\item If $\phi_F(m) \ge 2^{\lfloor \frac{m}2 \rfloor}$ then the
 family
\begin{quote}
 $F_2:$ all permutations on $[n]$ that differ from the identity by
 the transposition of any number of pairs of the form $\{2i, 2i+1\}$,
\end{quote}
 of size $2^{\lfloor \frac{n}2 \rfloor}$ and with $\phi_{F_2}(m) =
 2^{\lfloor \frac{m}2 \rfloor}$ shows that $|F|$ can be exponential in
 $n$.
\end{itemize}

 If $\phi_F(m)=m$ then $|I_F| = O(Ex(n,m,m-1))$, which is superlinear
 and $O(n^{1+\frac1{\lfloor \frac{m}2 \rfloor}})$~\cite{gst-ext98}. We
 have not found any example showing that $F$ could have superlinear
 size. The main transitions are summarized in
 Table~\ref{tab:transitionsP}.

\begin{table}[!htb]
\hspace{-1cm}
\begin{tabular}[5]{|c||c|c|c|c|}
\hline
 & $\phi_F(m) \le \lfloor \frac{m}2\rfloor$ & $\lfloor \frac{m}2 \rfloor < \phi_F(m) < 2\lfloor \frac{m}2 \rfloor$ & $2\lfloor \frac{m}2 \rfloor \le \phi_F(m) < 2^{\lfloor \frac{m}2 \rfloor}$ & $2^{\lfloor \frac{m}2 \rfloor} \le \phi_F(m)$ \\
\hline
\hline
$|F|$ & $\Theta(1)$ & $\Theta(n)$ & $\Omega(n)$ and $O\left(n^{2\lfloor \frac{m}2 -2 \rfloor}\right)$ & $\Omega(2^{\lfloor \frac{n}2 \rfloor})$\\
\hline
\end{tabular}
\caption{Maximum size of a family $F$ of permutations as a function of $\phi_F(m)$.\label{tab:transitionsP}}
\end{table}

\paragraph{\bf Exponential upper bounds.} 
 Raz~\cite{Raz00vc-dimensionof} proved that if $\phi_F(3) \le 5$ then
 $|F|$ has size at most exponential in $n$. The following simple
 observation derives a similar bounds for a few other values of
 $\phi_F(m)$.

\begin{lemma}
 If $|F|$ is at most exponential whenever $\phi_F(m-1) \le k-1$ then $|F|$
 is at most exponential whenever $\phi_F(m) \le k$.
\end{lemma}
\begin{proof}
 Let $T(n,m,k)$ denote the maximum size of a family $F$ such that
 $\phi_F(m) \le k$. Assume that $\phi_F(m) = k = \phi_F(m-1)$ as
 otherwise the statement trivially holds. Let $X \in \binom{[n]}{m-1}$
 such that $F_{|X} = \{\sigma_1, \ldots, \sigma_k\}$ has size $k$, and
 let:
\[ F_i = \{ \sigma \in F \mid \sigma_{|X} = \sigma_i\}.\]
 Observe that $F$ is the disjoint union of the $F_i$. Since
 $\phi_F(m)=k$, for any $e \in [n]\setminus X$ and for any $i=1,
 \ldots, k$, there exists a unique permutation in $(F_i)|_{X \cup
 \{e\}}$ that restricts to $\sigma_i$ on $X$. In other words, for
 every element in $[n]\setminus X$, the set $X \cup \{e\}$ appears in
 the same order in all permutations of $F_i$. It follows that
\[ |F_i| = |(F_i)_{|[n]\setminus X}|,\]
 that is, deleting $X$ does not decrease the size of each $F_i$
 considered individually -- although it may decrease the size of
 $F$. Now, let $G_i = (F_i)_{|[n]\setminus X}$ and consider the set
 system $([n]\setminus X,G_i)$. If $\phi_{G_i}(m-1)\le k-1$ then
 $|G_i| \le T(n-m+1,m-1,k-1)$, and otherwise $\phi_{G_i}(m-1) =k$ and
 we recurse. Altogether, we have the recursion
\[ T(n,m,k) \le k \max\left(T(n-m+1,m-1,k-1),T(n-m+1,m,k) \right),\]
 and it follows that if $T(n,m-1,k-1)$ is at most exponential, so is
 $T(n,m,k)$.
\end{proof}

 It then follows, with Raz's result, that $|F|$ is at most exponential
 whenever $\phi_F(m) \le m+2$.  Table~\ref{tab:smallvalsP} tabulates
 our results for small values of $m$ and $\phi_F(m)$, using the
 currently best known bounds on $Ex(n,k,\mu)$ we are aware
 of~\cite{gst-ext98}.

\begin{table}[!htb]
\hspace{-1.7cm}
\begin{tabular}[7]{|c||c|c|c|c|c|c|c|c|c|}
\hline
 & $k=2$ & $k=3$ & $k=4$ & $k=5$ & $k=6$ & $k=7$ & $k=8$ & $k=9$ & $k=10$\\
\hline
\hline
$m=2$ & $n!$ & - & - & - & - & - & - & - & -\\
\hline
$m=3$ & $2^{\Theta(n)}$ & $2^{\Theta(n)}$ & $2^{\Theta(n)}$ & $2^{\Theta(n)}$ & $n!$ & - & - &- &-\\
\hline
$m=4$ & $2$ & $\lfloor\frac{2n}3\rfloor+1$ & $2^{\Theta(n)}$ & $2^{\Theta(n)}$ & $2^{\Theta(n)}$ & $2^{\Omega(n)}$  & $2^{\Omega(n)}$ & $2^{\Omega(n)}$  & $2^{\Omega(n)}$\\
\hline
$m=5$ & $2$ & $\lfloor\frac{n}2\rfloor+1$ & $2^{\Theta(n)}$ & $2^{\Theta(n)}$ & $2^{\Theta(n)}$ & $2^{\Theta(n)}$  & $2^{\Omega(n)}$ & $2^{\Omega(n)}$  & $2^{\Omega(n)}$\\
\hline
$m=6$ & $2$ & $3$ & $\Theta(n)$ & $\Theta(n)$ & $O(n^3)$ & $O(n^3)$ & $2^{\Theta(n)}$ & $2^{\Omega(n)}$ & $2^{\Omega(n)}$\\
\hline
$m=7$ & $2$ & $3$ & $\Theta(n)$ & $\Theta(n)$ & $O(n^2)$ & $O(n^3)$ & $2^{\Theta(n)}$ & $2^{\Theta(n)}$ & $2^{\Omega(n)}$  \\
\hline
$m=8$ & $2$ & $3$ & $4$ & $\Theta(n)$ & $\Theta(n)$ & $\Theta(n)$ & $O(n^{21/8})$ & $O(n^{7/2})$ & $O(n^{7/2})$\\
\hline
$m=9$ & $2$ & $3$ & $4$ & $\Theta(n)$ & $\Theta(n)$ & $\Theta(n)$ & $O(n^2)$ & $O(n^{7/2})$ & $O(n^{7/2})$\\
\hline
$m=10$ & $2$ & $3$ & $4$ & $5$ & $\Theta(n)$ & $\Theta(n)$ & $\Theta(n)$ & $\Theta(n)$ & $O(n^{27/10})$\\
\hline
\end{tabular}
\caption{Maximum size of a family $F$ of permutations on $[n]$ with $\phi_F(m)=k$.\label{tab:smallvalsP}}
\end{table}

\section{Conclusion}

 A natural open question is the tightening of the bounds for both
 Questions~$1$ and~$2$. In particular, the first case where
 Lemma~\ref{lem:correspondance} no longer guarantees that the
 reduction from permutations to set systems leads to a ground set with
 linear size is $\phi_F(m)=m$; does that condition still imply that
 $|F|$ is $O(n)$ when $m$ is large enough?

\medskip

 Raz's generalization of the Stanley-Wilf conjecture, that is, whether
 $\phi_F(m)<m!$ implies that $|F|$ is exponential in $n$, also appears
 to be a challenging question. Can it be tackled by a
 ``normalization'' technique similar to Lemma~\ref{lem:normalization}?

\medskip

 A line intersecting a collection $C$ of pairwise disjoint convex sets
 in $\R^d$ induces two permutations, one reverse of the other,
 corresponding to the order in which each orientation of the line
 meets the set. The pair of these permutations is called a
 \emph{geometric permutation} of $C$. One of the main open questions
 in geometric transversal theory~\cite{w-httgt-04} is to bound the
 maximum number of geometric permutations of a collection of $n$
 pairwise disjoint sets in $\R^d$ (see for
 instance~\cite{as-gpiltfp-05,ak-mngpdtcs-06,cgn-gpdus-05,ss-ngp-03}). We
 can pick from each geometric permutation one of its elements so that
 the resulting family $F$ has the following property: if any $m$
 members of $C$ have at most $k$ distinct geometric permutations then
 $\phi_F(m-2) \le k$. One interesting question is whether bounds such
 as the one we obtained could lead to new results on the geometric
 permutation problem.

\section*{Acknowledgments}

The authors thank Olivier Devillers and Csaba T\'oth for their helpful comments.

\bibliographystyle{plain}
\bibliography{trans}
\end{document}